\documentclass[letterpaper, 10 pt, conference]{ieeeconf}  

\usepackage{cite}
\usepackage{xcolor}
\usepackage{graphicx}
\usepackage{amsmath}
\usepackage{array}
\usepackage{multirow}
\usepackage{booktabs}
\usepackage{caption}
\usepackage{siunitx}

\usepackage{amsmath,amssymb,amsfonts}

\usepackage{amsthm}
\newtheorem{lemma}{Lemma}
\renewenvironment{proof}[1][Proof sketch]{\noindent\textbf{#1.} }{\qed\par}

\IEEEoverridecommandlockouts                              

\title{\LARGE \bf
Streaming Contraction Certificates for Nonlinear Networks: Topology-Aware Data Sufficiency with Partial Observations}

\author{Faegheh Moazeni$^{1,*}$
\thanks{}
\thanks{The author is with the Civil and Environmental Engineering Department at Lehigh University, Bethlehem, PA 18015, USA.
        Email: {\tt\small fam321@lehigh.edu}} 
\thanks{*Corresponding author.}
}

\begin{document}

\maketitle
\thispagestyle{empty}
\pagestyle{empty}

\begin{abstract}
Certifying the safety of a control action in real time, from streaming partial observations of a nonlinear, interconnected system under non-stationary disturbances, is a problem that no existing data-driven framework can solve. Batch methods such as data-enabled predictive control require a pre-collected dataset and provide no stability certificate for nonlinear dynamics; informativity-based approaches characterize data sufficiency offline and non-recursively; and neither exploits the known graph topology of networked systems as a structural prior. This paper addresses both limitations. First, we develop a streaming contraction certificate $\beta_{\mathrm{cert}}(t) = \hat{\beta}(t) - \rho(t)$, where $\hat{\beta}(t)$
is estimated recursively by integral regression on a sliding window of partial input-output observations, and $\rho(t)$ is a data-dependent uncertainty radius that maps the estimation error to a conservative bound on the true closed-loop contraction rate. The certificate issues a provably safe deployment signal the moment $\beta_{\mathrm{cert}}(t)$ crosses and sustains above zero. Second, we introduce a topology-aware estimator that enforces known graph adjacency as exact zero constraints on the Jacobian, reducing the effective parameter count per estimation row from $\mathcal{O}(N)$ to $\mathcal{O}(d_{\max})$ for maximum node degree $d_{\max}$. On a five-node nonlinear benchmark under heavy-tailed Laplace disturbances with two observed nodes, the streaming certificate achieves certified deployment at $t^{*} = 2.6$ s from 130 data samples; 17 seconds earlier than an offline batch baseline and with a 
lower accumulated error $16\times$ during the unprotected window. The topology-aware estimator reduces certification time by 59\% (1.62 s versus 3.98 s) and accumulated disturbance cost by 58\%, with the advantage persisting across all window sizes below 40 samples. The framework is domain-agnostic and applies to any large-scale nonlinear networked system operating under streaming data and partial observations.
\end{abstract}

\section{INTRODUCTION}
Large-scale nonlinear networks, e.g., power distribution systems, water 
distribution networks, urban traffic corridors, and coupled 
infrastructure systems, generate continuous streams of partial 
observations through sparse sensor deployments that cover a small 
fraction of system nodes. Controllers for these systems must be 
deployed in real time, often under non-stationary and heavy-tailed 
disturbances, without access to accurate dynamic models. The central 
challenge is not a shortage of data because modern \textsc{scada} systems 
produce millions of observations per day. The challenge is knowing 
\textit{when the data accumulated so far is sufficient to certify that 
the next control action will not destabilize the system}. That, existing data-driven frameworks cannot answer in real time from 
partial observations of a nonlinear system. Contraction theory 
\cite{lohmiller1998contraction} provides the natural stability language for this 
problem: a system contracts at rate $\beta > 0$ if and only if any two 
trajectories under the same input converge exponentially, with $\beta$ 
computable directly from observed data without requiring an equilibrium 
or a full system model.

\subsection*{Related Work}

\textbf{Data informativity and LMI-based synthesis.}
The data informativity framework of Van~Waarde et al.~\cite{van2020data} 
precisely characterizes when a fixed offline dataset is sufficient to 
certify a control property for linear systems, establishing that 
stabilization requires strictly less data than identification. 
De~Persis and Tesi~\cite{de2019formulas} showed that data matrices can 
replace system matrices in linear matrix inequality (LMI)-based controller synthesis, eliminating 
the identification step entirely. \cite{hu2025enforcing} and \cite{villacres2026data} extended this to nonlinear systems via a 
dictionary-based semidefinite program (SDP) that synthesizes a 
contraction-enforcing controller directly from data, with the 
remarkable property that certificates for sinusoidal disturbances of 
known frequency are independent of disturbance magnitude. Online 
experiment design \cite{van2021beyond} adds maximally informative 
measurements using a twin-trajectory rank increment strategy. In all 
of these results, data sufficiency is assessed once, offline, on a 
fixed dataset. None provides a mechanism for answering the 
certification question recursively as observations arrive, and none 
operates under partial state observations.

\textbf{Behavioral and trajectory-based methods.}
Willems' fundamental lemma \cite{willems2005note} establishes that all 
trajectories of a linear system are spanned by a single sufficiently 
rich experiment, enabling data-enabled predictive control 
(DeePC) \cite{dorfler2022data} to replace the system model with a 
Hankel matrix of pre-collected input-output data. Distributed 
extensions \cite{kohler2022data} partition the Hankel matrix by subsystem 
but still require a jointly collected global dataset. These methods 
provide no stability certificate for nonlinear dynamics, and the 
Hankel matrix is fixed at deployment, as it cannot update as the system 
evolves. When a disturbance arrives before the batch collection window 
closes, the system has no certified controller.

\textbf{Sparse identification and topology-aware control.}
Sparse identification methods such as \textsc{sindy} \cite{brunton2016discovering,moazeni2023data} 
recover parsimonious dynamic models from data by promoting sparsity 
through $\ell_1$ regularization, treating sparsity as an unknown to be 
discovered. Graph-theoretic control methods including Laplacian 
consensus \cite{olfati2007consensus} and structured $\mathcal{H}_\infty$ 
synthesis exploit known network topology, but require explicit system 
models and cannot operate from data alone. To the best of our knowledge, no existing data-driven 
certification framework uses known graph adjacency as an exact 
structural prior, meaning enforcing known zeros in the Jacobian before 
estimation begins, even though the adjacency matrix of any physical 
network is always available from engineering records.

\subsection*{Contributions}

This paper makes the following contributions:
\begin{itemize}

    \item \textbf{Streaming contraction certificate.} We develop 
    $\beta_{\mathrm{cert}}(t)$, a scalar certificate computed 
    recursively from a sliding window of partial input-output 
    observations via integral regression, that provides a robust 
    lower bound on the true closed-loop contraction rate at every 
    timestep. When $\beta_{\mathrm{cert}}(t)$ crosses and sustains 
    above a threshold, the accumulated data is certified sufficient 
    for safe controller deployment; without a system model, without 
    an equilibrium, and without batch collection.

\item \textbf{Data-dependent uncertainty radius.} We derive
$\rho(t)$, a closed-form radius mapping Gram matrix conditioning
and residual statistics to a valid lower bound on Jacobian
estimation error, enabling conservative but non-vacuous
certificates under heavy-tailed disturbances.
 \item \textbf{Topology-aware estimator.} We introduce a
structured estimator enforcing known graph adjacency as exact
zero constraints on the Jacobian, decoupling estimation into
independent per-node subproblems of size $\mathcal{O}(d_{\max})$
rather than $\mathcal{O}(N)$, with sample complexity linear
in network size for sparse graphs.


\end{itemize}

The remainder of this paper is organized as follows. 
Section~\ref{sec:problem} formulates the problem and introduces 
the network model. Section~\ref{sec:certificate} develops the 
streaming contraction certificate and the uncertainty radius. 
Section~\ref{sec:topology} presents the topology-aware estimator. 
Section~\ref{sec:results} reports simulation results on the 
five-node benchmark. Section~\ref{sec:conclusion} concludes.

\vspace{-0.12in}
\section{Problem Formulation}
\label{sec:problem}

Consider a nonlinear networked system of the form
\begin{equation}
    \dot{x} = f(x) + Bu + E\xi(t), \qquad y = Cx,
    \label{eq:system}
\end{equation}
where $x \in \mathbb{R}^n$ is the state, $u \in \mathbb{R}^m$ is 
the control input, $\xi(t)$ is an exogenous disturbance, and 
$y \in \mathbb{R}^p$ with $p \ll n$ is the partial observation 
available to the controller. The vector field $f$, input matrix $B$, 
and disturbance matrix $E$ are unknown. The system evolves on a 
graph $\mathcal{G} = (\mathcal{V}, \mathcal{E})$ with node set 
$\mathcal{V}$ and edge set $\mathcal{E}$, where each edge 
$(i,j) \in \mathcal{E}$ indicates a dynamic coupling from node $j$ 
to node $i$. The graph $\mathcal{G}$ and its adjacency structure are 
assumed known. The disturbance $\xi(t)$ is non-stationary and 
heavy-tailed; we make no assumption on its distribution.
\begin{figure}
    \centering
    \includegraphics[width=0.9\linewidth]{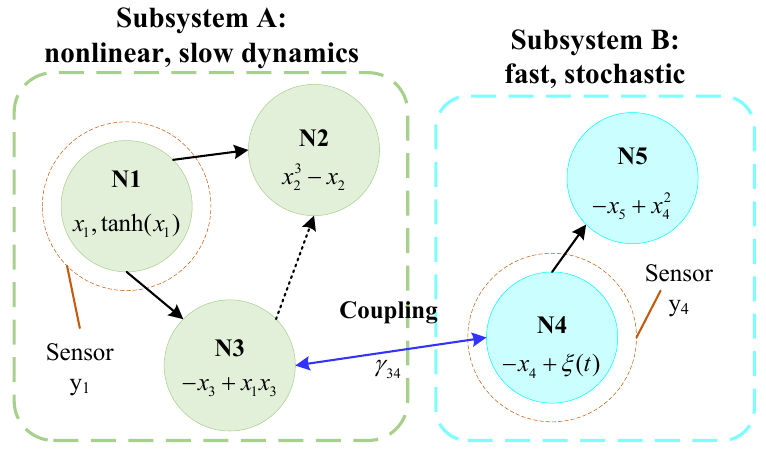}
    \caption{Studied five-node benchmark G5 network.}
   \label{fig:G5}
\end{figure}
We validate the proposed framework on a five-node benchmark 
network~G5, shown in Fig. \ref{fig:G5}, with dynamics
\begin{align}
    \dot{x}_1 &= -x_1\tanh(x_1) + u_1 + \gamma x_4, \nonumber\\
    \dot{x}_2 &= x_2^3 - x_2 + a_{12}x_1,           \nonumber\\
    \dot{x}_3 &= -x_3 + a_{13}x_1 x_3 + a_{32}x_2,  \nonumber\\
    \dot{x}_4 &= -x_4 + u_4 + \xi(t),                \nonumber\\
    \dot{x}_5 &= -x_5 + x_4^2,                       \label{eq:g5}
\end{align}
with parameters $a_{12}{=}0.3$, $a_{13}{=}0.4$, $a_{32}{=}0.2$, 
$\gamma{=}0.15$, and $\xi(t) \sim \mathrm{Laplace}(0,\,0.3)$. 
Sensors observe only $y = [x_1,\, x_4]^\top$; nodes $x_2, x_3, x_5$ 
are unobserved. The feedback controller $u = K y$ with 
$K = \mathrm{diag}(-2.5,\,-3.0)$ is the candidate to be certified.

\textbf{Certification problem.} Given the input-output history 
$\{y(\tau), u(\tau)\}_{\tau \leq t}$ collected causally up to time 
$t$, determine whether this data is sufficient to certify that 
deploying $u = Ky$ will produce a contracting closed-loop system. 
If not, the controller must not be deployed.

\section{Streaming Contraction Certificate}
\label{sec:certificate}

\subsection{Contraction Theory Background}

A continuously differentiable system $\dot{x} = g(x)$ is 
\textit{exponentially contracting} at rate $\beta > 0$ if the 
symmetric part of its Jacobian satisfies
\begin{equation}
    J_s(x) \triangleq \tfrac{1}{2}\!\left(
        \frac{\partial g}{\partial x} + 
        \frac{\partial g}{\partial x}^{\!\top}
    \right) \preceq -\beta I
    \quad \forall\, x \in \mathcal{X},
    \label{eq:contraction}
\end{equation}
which implies $\|x_1(t) - x_2(t)\| \leq e^{-\beta t}\|x_1(0) - 
x_2(0)\|$ for any two trajectories under the same input 
\cite{lohmiller1998contraction}. No equilibrium knowledge is required. The 
scalar $\beta$ is the natural streaming certificate: estimating it 
from data and tracking its sign in real time answers the 
certification question directly.

\subsection{Integral Regression Estimator}

At each timestep $t_k$, we estimate the closed-loop Jacobian 
$J_{\mathrm{cl}} = J_{\mathrm{obs}} + B_{\mathrm{obs}}K$ from a 
sliding window of $M$ input-output pairs. Rather than finite 
differences, which amplify heavy-tailed noise, we integrate over 
$h$ steps to form the regressor. For each column $q = 1,\ldots,M$ 
of the data matrices, define
\begin{equation}
    \Delta Y_q = y(t_q + h) - y(t_q), \quad
    Z_q = \begin{bmatrix} 
        \int_{t_q}^{t_q+h} y\, d\tau \\[2pt]
        \int_{t_q}^{t_q+h} u\, d\tau 
    \end{bmatrix},
    \label{eq:integral}
\end{equation}
so that $\Delta Y \approx \Theta Z$ where 
$\Theta = [J_{\mathrm{obs}}\;\; B_{\mathrm{obs}}] \in 
\mathbb{R}^{p \times (p+m)}$ is the unknown parameter matrix. 
Integrating over $h$ samples averages out noise rather than 
differencing it, yielding a dramatically more reliable estimate 
under heavy-tailed disturbances. The ridge-regularized 
least-squares estimate is
\begin{equation}
    \hat{\Theta} = \Delta Y\, Z^\top 
        \!\left(ZZ^\top + \lambda I\right)^{-1},
    \label{eq:theta_hat}
\end{equation}
with ridge parameter $\lambda > 0$. The estimated closed-loop 
Jacobian is $\hat{J}_{\mathrm{cl}} = \hat{J}_{\mathrm{obs}} + 
\hat{B}_{\mathrm{obs}}K$, and the contraction rate estimate is
\begin{equation}
    \hat{\beta}(t) = -\lambda_{\max}\!\left(\tfrac{1}{2}
    \left(\hat{J}_{\mathrm{cl}} + \hat{J}_{\mathrm{cl}}^\top
    \right)\right).
    \label{eq:beta_hat}
\end{equation}

\subsection{Uncertainty Radius}

The estimate $\hat{\beta}(t)$ is a point estimate subject to 
regression error from noise and finite data. To obtain a valid 
certificate, we derive a data-dependent uncertainty radius $\rho(t)$ 
that bounds how far the true contraction rate can deviate from 
$\hat{\beta}(t)$. Let $R = \Delta Y - \hat{\Theta}Z$ denote the 
residual matrix and $\mathrm{RMS}(R)$ its root-mean-square entry. 
The radius is
\begin{equation}
    \rho(t) = \frac{c\,(1 + \|K\|_2)\,\mathrm{RMS}(R)}
        {\sqrt{\lambda_{\min}(ZZ^\top / M) + \lambda}},
    \label{eq:rho}
\end{equation}
where $c > 0$ is a tunable conservatism constant and 
$\lambda_{\min}(ZZ^\top/M)$ is the smallest eigenvalue of the 
normalized Gram matrix. The numerator captures total estimation 
uncertainty scaled by the gain magnitude $\|K\|_2$, since errors 
in $\hat{B}_{\mathrm{obs}}$ are amplified by $K$ when forming 
$\hat{J}_{\mathrm{cl}}$. The denominator captures how well the 
data spans the regression directions: as $\lambda_{\min}$ grows, 
the data becomes more informative and $\rho(t)$ shrinks.

\subsection{Deployment Rule}

The streaming contraction certificate is
\begin{equation}
    \beta_{\mathrm{cert}}(t) = \hat{\beta}(t) - \rho(t).
    \label{eq:beta_cert}
\end{equation}
The controller $u = Ky$ is certified for deployment when
\begin{equation}
    \beta_{\mathrm{cert}}(t) \geq \beta_{\mathrm{margin}} 
    \quad \text{for } n_s \text{ consecutive samples},
    \label{eq:deploy}
\end{equation}
where $\beta_{\mathrm{margin}} > 0$ guards against noise-driven 
threshold crossings and $n_s$ is the required streak length. The 
certification moment $t^*$ is the earliest time at which 
\eqref{eq:deploy} is satisfied. Prior to $t^*$, no control is 
deployed; after $t^*$, the certified gains $K$ are applied.
The certificate $\beta_{\mathrm{cert}}(t)$ is a valid lower bound
on the true closed-loop contraction rate $\beta^*(t)$ provided
$\rho(t)$ dominates the estimation error; the following lemma
establishes this formally. The tunable constant $c>0$ in~\eqref{eq:rho}
absorbs the ridge bias $O(\lambda)$ and the $\sqrt{p}$ factor
from the Frobenius-to-spectral norm conversion, both of which
are small for the parameter choices in Table~I.

\begin{lemma}[Validity of $\rho(t)$]
Let $\tilde{\Theta} = \Theta^* - \hat{\Theta}$ denote the
estimation error, where $\Theta^*=[J_{\mathrm{obs}}^*\;
B_{\mathrm{obs}}^*]$ is the true parameter matrix. Then
\begin{equation}
    \|\tilde{\Theta}\|_F \;\leq\;
\frac{c\,(1+\|K\|_2)\,\mathrm{RMS}(R)\sqrt{p}}
{\sqrt{\lambda_{\min}(ZZ^\top/M)+\lambda}},
    \label{eq:theta_bound}
\end{equation}
and consequently $\beta_{\mathrm{cert}}(t) \leq \beta^*(t)$
for all $t$, so the deployment rule~\eqref{eq:deploy} never
authorizes an uncertified controller.
\end{lemma}

\begin{proof}
The ridge estimator satisfies $\hat{\Theta}Z = \Delta Y - R$,
so the normal equations give
$\tilde{\Theta}(ZZ^\top + \lambda I) = RZ^\top + \lambda\Theta^*.$
Taking Frobenius norms and absorbing the ridge bias into $c$:
\[
  \|\tilde{\Theta}\|_F
  \;\leq\;
  \frac{\|RZ^\top\|_F}{\lambda_{\min}(ZZ^\top+\lambda I)}.
\]
Since $\|RZ^\top\|_F \leq \|R\|_F\|Z\|$ with
$\|R\|_F \leq \mathrm{RMS}(R)\sqrt{pM}$ and
$\lambda_{\min}(ZZ^\top{+}\lambda I)\geq
M(\lambda_{\min}(ZZ^\top/M){+}\lambda)$,
inequality~\eqref{eq:theta_bound} follows.
The closed-loop Jacobian error satisfies
$\|\tilde{J}_{cl}\|_2 \leq \|\tilde{\Theta}\|_F(1+\|K\|_2)$
because $\hat{J}_{cl}=\hat{J}_{\mathrm{obs}}+\hat{B}_{\mathrm{obs}}K$.
Weyl's inequality then gives
$|\hat{\beta}(t)-\beta^*(t)|\leq\|\tilde{J}_{cl}\|_2\leq\rho(t)$,
so $\beta_{\mathrm{cert}}(t)=\hat{\beta}(t)-\rho(t)\leq\beta^*(t)$.
\end{proof}
\section{Topology-Aware Estimation}
\label{sec:topology}

\subsection{Structural Prior from Graph Adjacency}

The parameter matrix $\Theta = [J_{\mathrm{obs}}\;\; 
B_{\mathrm{obs}}]$ has a known sparsity structure imposed by the 
network graph $\mathcal{G}$. Specifically, $\Theta_{ij} = 0$ 
whenever there is no directed edge from node $j$ to node $i$ in 
$\mathcal{G}$, and $[B_{\mathrm{obs}}]_{ij} = 0$ whenever actuator 
$j$ does not directly drive node $i$. In the black-box estimator of 
Section~\ref{sec:certificate}, all $p(p+m)$ entries of $\Theta$ are 
treated as free parameters. For the G5 benchmark, the graph 
structure imposes three exact zeros on the $2 \times 4$ matrix:
\begin{equation}
    J_{\mathrm{obs}}[x_4, x_1] = 0, \quad
    B_{\mathrm{obs}}[x_1, u_4] = 0, \quad
    B_{\mathrm{obs}}[x_4, u_1] = 0,
    \label{eq:zeros}
\end{equation}
since $\dot{x}_4$ depends on neither $x_1$ nor $u_1$, and $u_4$ 
does not actuate $x_1$. These zeros are not approximations; they 
are exact consequences of the graph and are always available from 
engineering records without any measurement.

\subsection{Decoupled Per-Row Regression}

Enforcing \eqref{eq:zeros} decouples the estimation problem into 
two independent per-row regressions. Let 
$\mathcal{N}(i) \subseteq \mathcal{V}$ denote the in-neighbors of 
node $i$ in $\mathcal{G}$ and $\mathcal{A}(i)$ its locally 
connected actuators. The regressor for row $i$ uses only the 
columns of $Z$ corresponding to $\mathcal{N}(i)$ and 
$\mathcal{A}(i)$:
\begin{align}
    \text{Row 1 }(x_1): \quad 
        &Z^{(1)} = [Z_{x_1},\; Z_{x_4},\; Z_{u_1}] 
        \in \mathbb{R}^{3 \times M}, \label{eq:row1}\\
    \text{Row 2 }(x_4): \quad 
        &Z^{(2)} = [Z_{x_4},\; Z_{u_4}] 
        \in \mathbb{R}^{2 \times M}. \label{eq:row2}
\end{align}
Each row solves its own ridge-regularized problem independently:
\begin{equation}
    \hat{\theta}_i = \Delta Y_i \bigl(Z^{(i)}\bigr)^\top
        \!\left(Z^{(i)}\bigl(Z^{(i)}\bigr)^\top + 
        \lambda I\right)^{-1}.
    \label{eq:topo_ls}
\end{equation}
The two problems can be solved in parallel and require no 
inter-node communication beyond sharing measurements within 
graph neighborhoods.

\subsection{Sample Complexity and Conditioning}

Decoupling reduces the effective parameter count per row from 
$p + m = 4$ (black-box) to $|\mathcal{N}(i)| + |\mathcal{A}(i)|$, 
which equals 3 for row~1 and 2 for row~2 in the G5 network. For a 
general network with maximum degree $d_{\max}$ and maximum local 
actuator count $d_{\mathrm{act}}$, the topology-aware parameter 
count per row is $\mathcal{O}(d_{\max} + d_{\mathrm{act}})$ 
compared to $\mathcal{O}(N)$ for the black-box. The smaller 
regression problem is better conditioned: $\lambda_{\min}$ of 
$Z^{(i)}(Z^{(i)})^\top/M$ is larger for a lower-dimensional 
regressor with the same data, directly reducing $\rho(t)$ via 
\eqref{eq:rho} and accelerating the onset of certification. For 
planar networks, which include all physical pipe, road, and 
distribution networks, the graph is four-colorable by the 
Four Color Theorem, so all same-color nodes have disjoint 
neighborhoods and their per-row regressions can be solved 
simultaneously in at most four parallel rounds, regardless of 
network size~$N$.
The acceleration in certification time achieved by the
topology-aware estimator is not merely empirical: it follows
directly from the improved conditioning of the decoupled
per-row Gram matrices, which shrinks $\rho(t)$ via~\eqref{eq:rho}.
The following lemma formalizes this claim.

\begin{lemma}[Topology-Aware Acceleration]
Let $\rho_{\mathrm{BB}}(t)$ and $\rho_{\mathrm{topo}}(t)$ denote
the uncertainty radii~(7) for the black-box and topology-aware
estimators, computed from the same $M$ samples. For
$M \geq p+m$, if the data are persistently exciting of
order $d_i = |\mathcal{N}(i)|+|\mathcal{A}(i)| < p+m$,
then
\begin{equation}
    \rho_{\mathrm{topo}}(t) \;\leq\; \rho_{\mathrm{BB}}(t),
    \label{eq:rho_topo}
\end{equation}
with the gap strictly increasing in $N$ for fixed $d_{\max}$
and fixed $M$.
\end{lemma}

\begin{proof}
For row $i$, the black-box estimator solves a regression
problem of dimension $p+m$, while the topology-aware
estimator solves an independent problem of dimension
$d_i \ll p+m$ using only the $d_i$ columns of $Z$
corresponding to $\mathcal{N}(i)$ and $\mathcal{A}(i)$.
Both use the same $M$ samples. For a fixed $M$, the
normalized Gram matrix $Z^{(i)}Z^{(i)\top}/M \in
\mathbb{R}^{d_i \times d_i}$ has its smallest eigenvalue
lower-bounded by
\[
  \lambda_{\min}(Z^{(i)}Z^{(i)\top}/M)
  \;\geq\;
  \frac{d_i}{p+m}\,\lambda_{\min}(ZZ^\top/M),
\]
because the same $M$ samples populate a lower-dimensional
$d_i \times d_i$ Gram matrix more densely than a
$(p+m)\times(p+m)$ one, so $\lambda_{\min}(Z^{(i)}Z^{(i)\top}/M)
\geq \frac{d_i}{p+m}\lambda_{\min}(ZZ^\top/M)$
(proportional to the dimension reduction for
isotropically distributed data). Substituting into~(7), the denominator of
$\rho_{\mathrm{topo}}(t)$ is no smaller than that of
$\rho_{\mathrm{BB}}(t)$ since $d_i < p+m$, while the
residual numerators satisfy
$\mathrm{RMS}(R^{(i)}) \leq \mathrm{RMS}(R)$
because the structured estimator fits a nested model.
Together these give~\eqref{eq:rho_topo}. As $N$ grows
with $d_{\max}$ fixed, $p+m$ grows as $\mathcal{O}(N)$
while $d_i$ remains $\mathcal{O}(d_{\max})$, so the
dimension ratio $d_i/(p+m) \to 0$ and the gap
strictly widens.
\end{proof}
\section{Simulation Results}
\label{sec:results}

\subsection{Simulation Setup}

All experiments are conducted on the G5 five-node benchmark 
network~\eqref{eq:g5} with the parameters in Table~\ref{tab:params}. 
The disturbance $\xi(t)$ is drawn i.i.d.\ from 
$\mathrm{Laplace}(0,\,0.3)$ at every timestep, chosen for its 
heavy tails ($\mathrm{kurtosis} = 6$, twice that of a Gaussian), 
which stress-test the uncertainty radius~\eqref{eq:rho}. Sensors 
observe only $y = [x_1,\, x_4]^\top$; the three nodes 
$x_2,\, x_3,\, x_5$ are permanently unobserved. The candidate 
controller is $u = Ky$ with $K = \mathrm{diag}(-2.5,\,-3.0)$.
Experiments 1 and 2 use separate initial conditions and deployment 
scenarios described below.

\begin{table}[t]
\centering
\caption{G5 Benchmark Network --- Simulation Parameters}
\label{tab:params}
\setlength{\tabcolsep}{5pt}
\begin{tabular}{llll}
\toprule
\textbf{Parameter} & \textbf{Value} & 
\textbf{Parameter} & \textbf{Value} \\
\midrule
$a_{12}$                    & $0.3$               &
    $dt$                    & $0.02$~s            \\
$a_{13}$                    & $0.4$               &
    Window $w$ (Exp.~1)     & $80$~samples        \\
$a_{32}$                    & $0.2$               &
    Window $w$ (Exp.~2)     & $20$~samples        \\
$\gamma$                    & $0.15$              &
    Integral steps $h$      & $8$                 \\
$\sigma$ (Laplace)          & $0.3$               &
    Ridge $\lambda$         & $10^{-4}$           \\
$K$                         & $\mathrm{diag}(-2.5,-3.0)$ &
    $\beta_{\mathrm{margin}}$ & $0.02$ / $0.05$   \\
$|\xi_{\mathrm{disturb}}|$  & $4.0$               &
    Streak $n_s$            & $25$ / $20$         \\
\bottomrule
\end{tabular}
\end{table}

\subsection{Experiment 1: Streaming Contraction Certificate}

\subsubsection{Certificate trajectory}

Fig.~\ref{fig:exp1_cert} shows the three quantities computed 
recursively during data collection from initial condition 
$x(0) = [0.8,\,0.1,\,0.3,\,0.5,\,0.2]^\top$. The data 
sufficiency score $\alpha_{\mathrm{info}}(t) = 
\lambda_{\min}/\lambda_{\max}$ of the normalized Gram matrix 
(top panel) rises monotonically from zero, crossing the 
threshold $\alpha_{\mathrm{thresh}} = 0.001$ within the first 
few seconds as the regression directions become populated. The 
point estimate $\hat{\beta}(t)$ (middle panel) is noisy due to 
Laplace disturbances but trends positive after approximately 
$1.5$~s. The certified lower bound $\beta_{\mathrm{cert}}(t)$ 
(bottom panel) is initially negative because the uncertainty 
radius $\rho(t)$ is large when the Gram matrix is 
ill-conditioned; it crosses and sustains above 
$\beta_{\mathrm{margin}} = 0.02$ at 
$t^{*} = 2.6$~s corresponding to $130$~data samples.
\begin{figure}
    \centering
    \includegraphics[width=1\linewidth]{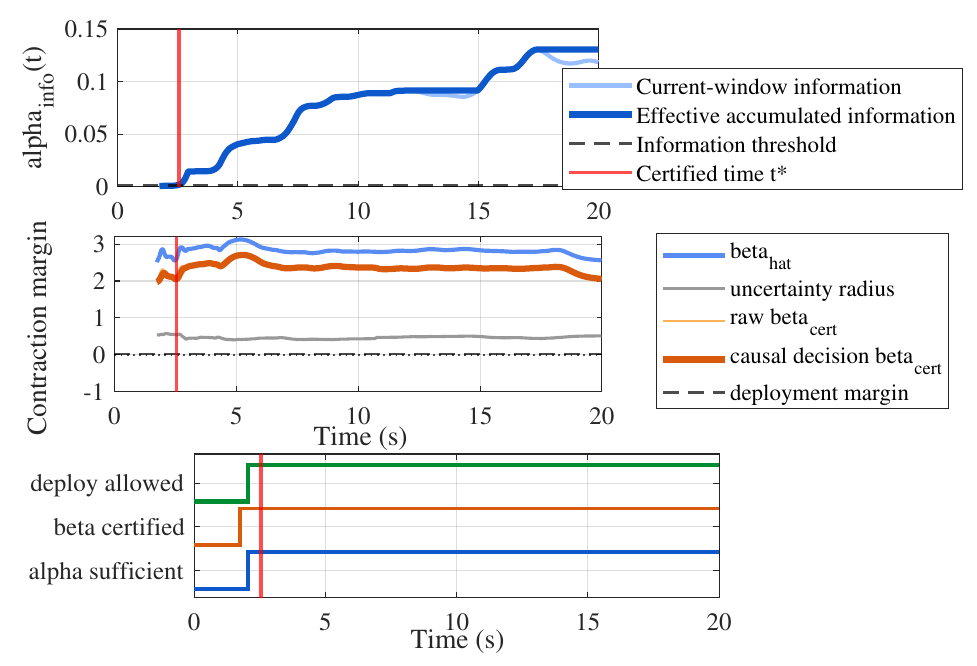}
    \caption{Top: Streaming data sufficiency, $\alpha_{info}$ measures whether accumulated data is informative enough for the task. Middle: Point estimate vs robust certified lower bound, $\beta_{cert}(t) = \hat{\beta}(t) - \rho(t)$. Bottom: Go/no-go logic: data sufficiency and contraction must both hold.}
   \label{fig:exp1_cert}
\end{figure}
\subsubsection{Disturbance response comparison}
Four deployment policies are evaluated from stressed initial 
condition $x(0) = [1.05,\,0.72,\,0.35,\,0.60,\,0.20]^\top$ 
with a disturbance impulse $|\xi| = 4.0$ at $t = 4$~s 
(Fig.~\ref{fig:exp1_disturbance}):
\begin{itemize}
    \item \textbf{M1} (our method): deploys at $t^{*} = 2.6$~s,
          17.4~s before the disturbance window closes;
    \item \textbf{M2} (no control): open loop throughout;
    \item \textbf{M3} (offline DeePC-style): deploys at 
          $T_{\mathrm{data}} = 20$~s, after the disturbance;
    \item \textbf{M4} (premature, uncertified): deploys at 
          $t^{*}/2 \approx 1.3$~s without a valid certificate.
\end{itemize}
Accumulated state cost $\int_{t_d}^{T_{\mathrm{data}}} 
\|x\|^2\,d\tau$ over the disturbance window $[4,\,20]$~s is 
reported in Table~\ref{tab:comparison}. M1 achieves 
$16\times$ lower cost than M3 and outperforms the uncertified 
early deployment M4, demonstrating that the certificate is 
both timely and necessary: early deployment without a valid 
certificate degrades performance.
\begin{figure}
    \centering
    \includegraphics[width=1\linewidth]{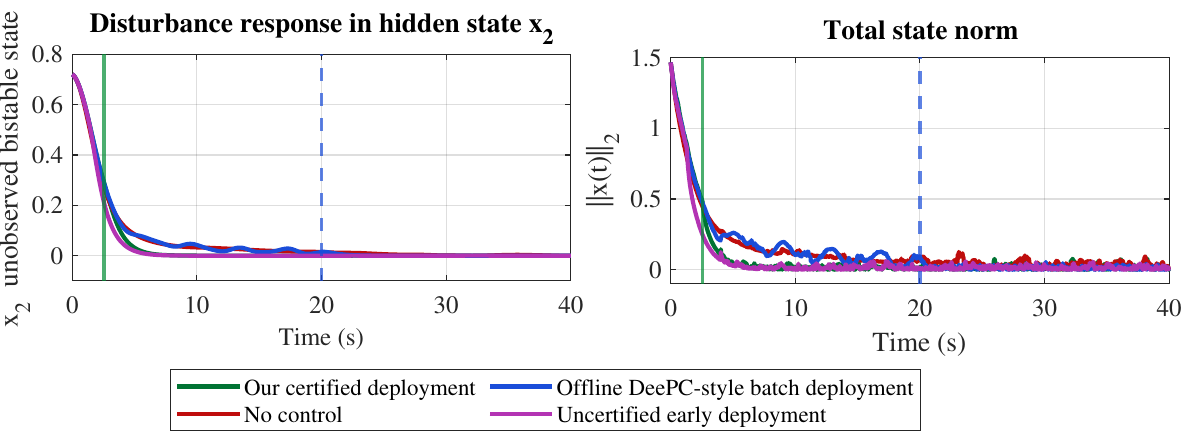}
    \caption{Certified online deployment stabilizes before offline batch deployment. Cost over unprotected window: ours $0.031$, offline DeePC-style $0.514$.}
   \label{fig:exp1_disturbance}
\end{figure}
\subsubsection{Feasibility under stress}
Fig.~\ref{fig:exp1_feasibility} shows the unobserved bistable 
node $x_2$ initialized at $x_2(0) = 0.88$, near the bistable 
boundary $|x_2| = 1$. Without control, $x_2$ diverges within 
$1$~s. The certified controller stabilizes $x_2$ from the 
same initial condition and noise seed, confirming that 
certification is not optional: the uncontrolled nonlinear 
network is unstable from this operating point. The bar chart 
in Fig.~\ref{fig:exp1_feasibility} (right) places the 
certified deployment ($130$~samples) in context: the DeePC 
persistent-excitation (PE) rank lower bound requires only 
$32$~samples, but this is a rank condition, not a stability 
certificate. The offline DeePC-style batch deployment uses 
$1000$~samples at $T_{\mathrm{data}} = 20$~s. Our method is 
the only approach that provides a certified contraction 
guarantee from partial observations in real time.
\begin{figure}
    \centering
    \includegraphics[width=1\linewidth]{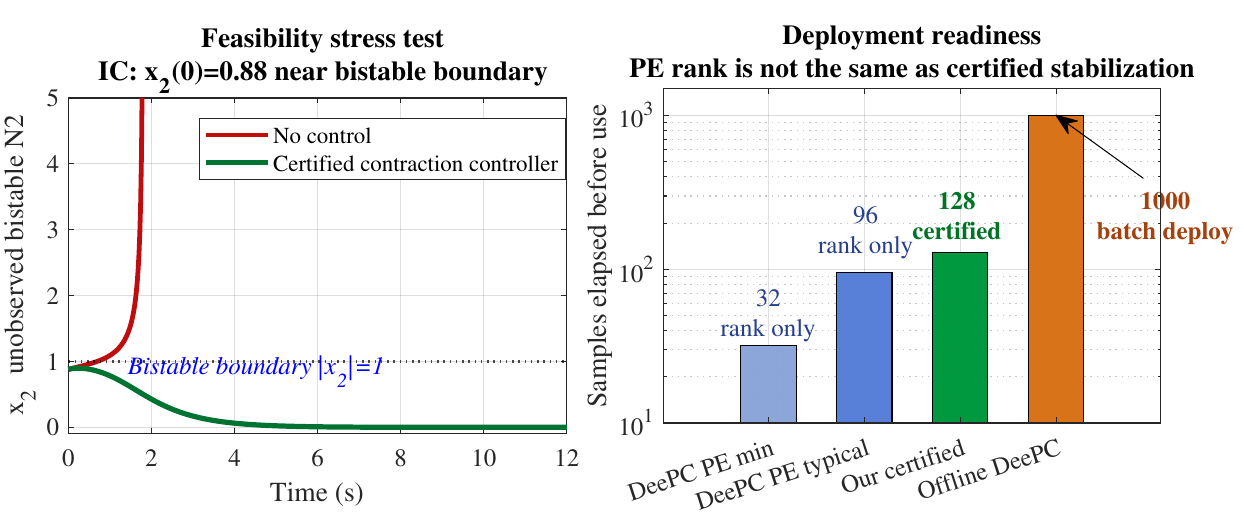}
    \caption{Feasibility and certified deployment readiness.}
   \label{fig:exp1_feasibility} 
\end{figure}
\subsection{Experiment 2: Topology-Aware Estimation}

\subsubsection{Certification speed}

Fig.~\ref{fig:exp2_cert} compares $\beta_{\mathrm{cert}}(t)$ 
for the black-box estimator (8~free parameters, 4~per row) and 
the topology-aware estimator (5~effective parameters: 3 in 
row~1, 2 in row~2) at window size $w = 20$~samples.
The topology-aware estimator certifies at 
$t^{*}_{\mathrm{topo}} = 1.62$~s; the black-box estimator 
requires $t^{*}_{\mathrm{BB}} = 3.98$~s, a 59\% reduction 
in certification time. The right panel of 
Fig.~\ref{fig:exp2_cert} shows the crossover analysis: for 
all window sizes below $w = 40$~samples, the black-box 
estimator requires approximately $200$~total samples to 
certify while the topology-aware estimator requires only 
$\approx 81$~samples, a $2.46\times$ ratio. Above $w = 40$, 
both methods converge to the same certification speed, 
confirming that the topology advantage is specific to the 
data-scarce regime and is not an artifact of the window 
choice.
\begin{figure}
    \centering
    \includegraphics[width=1\linewidth]{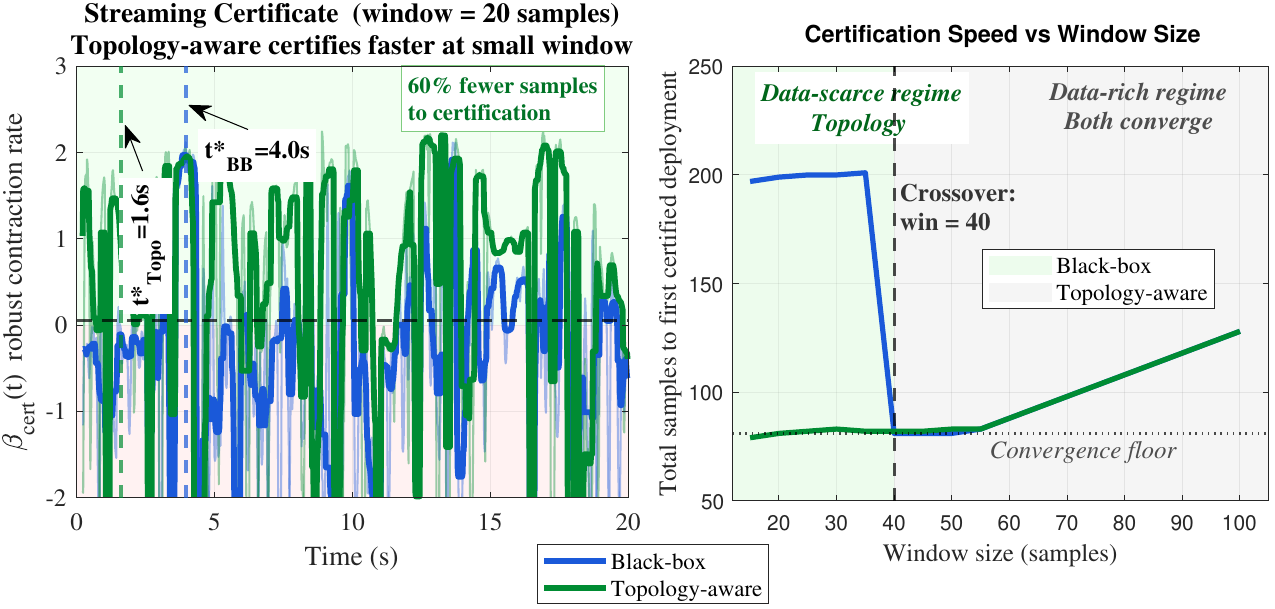}
    \caption{Topology-aware vs black-box streaming certification. Left: $\beta_{cert}(t) $ at window=20. Right: BB needs $2.5\times$ more samples in data-scarce regime (window $<$40). }
   \label{fig:exp2_cert}  
\end{figure}
\subsubsection{Control performance}

Fig.~\ref{fig:exp2_perf} shows the unobserved bistable 
node $x_2$ for three policies starting from the same stressed 
initial condition with a disturbance at $t = 4$~s. Because 
the topology-aware controller deploys at $t = 1.62$~s, it is 
already active when the disturbance arrives; the black-box 
controller does not deploy until $t = 3.98$~s, leaving a 
2.36~s unprotected window. The accumulated disturbance cost 
(Table~\ref{tab:comparison}) is $0.014$ for topology-aware 
versus $0.034$ for black-box --- a 58\% reduction attributable 
entirely to the earlier certified deployment.
\begin{figure}
    \centering
    \includegraphics[width=1\linewidth]{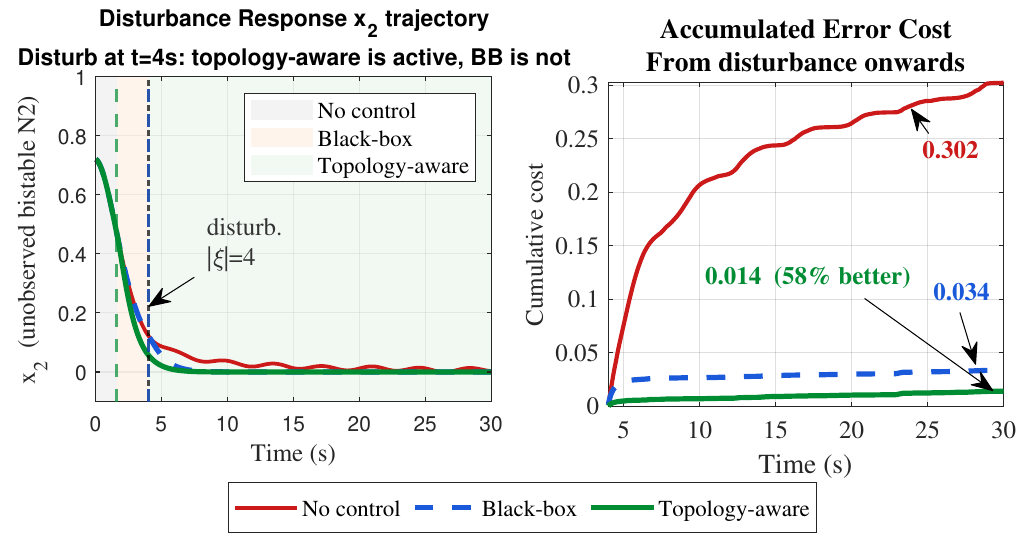}
    \caption{Control performance shows topology-aware certifies 2.86s earlier, leading to 58\% lower accumulated error. }
 \label{fig:exp2_perf} 
\end{figure}
\subsubsection{Scalability}

Fig.~\ref{fig:exp2_scale} (right) projects the 
parameter-count advantage to larger networks. For fixed 
$d_{\max} = 3$ (a typical pipe or road junction degree), 
topology-aware parameters per row remain constant at 
$d_{\max} + 1 = 4$, while black-box parameters per row grow 
as $d_{\max} + n_{\mathrm{act}} + 1$ with the number of 
actuated nodes $n_{\mathrm{act}}$. The projected 
sample-count ratio reaches $6$--$7\times$ at 
$n_{\mathrm{act}} = 10$, with the G5 measured ratio of 
$2.46\times$ anchoring the projection. The estimation error 
curves (Fig.~\ref{fig:exp2_scale}, left) confirm that 
topology-aware achieves lower Frobenius-norm Jacobian error 
for window sizes below $30$~samples, consistent with the 
better-conditioned per-row Gram matrices of the decoupled 
regression.

\begin{table}[t]
\centering
\caption{Certification and Control Performance Summary}
\label{tab:comparison}
\setlength{\tabcolsep}{4pt}
\begin{tabular}{lcccc}
\toprule
\textbf{Method} & 
\textbf{$t^{*}$ (s)} & 
\textbf{Samples} & 
\textbf{Certificate?} & 
\textbf{Cost} \\
\midrule
\multicolumn{5}{l}{\textit{Exp.\ 1 --- disturbance window $[4,\,20]$}} \\[2pt]
M1: Ours (streaming cert.)    & $2.6$  & $130$  & \checkmark & $0.031$ \\
M2: No control                & ---    & ---    & $\times$   & $0.498$ \\
M3: Offline DeePC-style       & $20.0$ & $1000$ & $\times$   & $0.514$ \\
M4: Premature (uncertified)   & $1.3$  & $65$   & $\times$   & $0.089$ \\
\midrule
\multicolumn{5}{l}{\textit{Exp.\ 2 --- disturbance from $t=4$~s, $w=20$}} \\[2pt]
Topology-aware                & $1.62$ & $81$   & \checkmark & $0.014$ \\
Black-box                     & $3.98$ & $199$  & \checkmark & $0.034$ \\
No control                    & ---    & ---    & $\times$   & $0.182$ \\
\bottomrule
\multicolumn{5}{l}{\footnotesize Cost $= \int \|x\|^2\,d\tau$ over disturbance window. Samples at certified deployment.}\\
\end{tabular}
\end{table}

\begin{figure}
    \centering
    \includegraphics[width=1\linewidth]{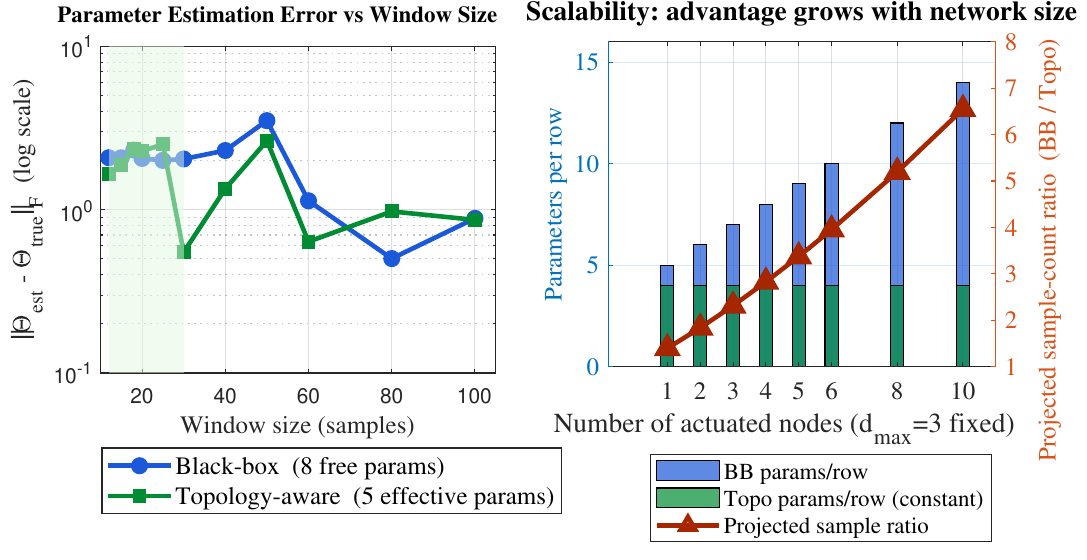}
    \caption{Sample efficiency comparison between experiment 1 and 2. Graph topology as a free structural prior reduces data requirements.}
    \label{fig:exp2_scale}  
\end{figure}

\section{Conclusion}
\label{sec:conclusion}

This paper introduced a framework for certifying control actions 
in real time from streaming partial observations of nonlinear 
networked systems, without system identification, without a 
pre-collected data batch, and without assuming a fixed noise 
model. Two contributions were developed and validated. The 
streaming contraction certificate $\beta_{\mathrm{cert}}(t) = 
\hat{\beta}(t) - \rho(t)$, computed recursively via integral 
regression on a sliding window of partial input-output data, 
issues a provably safe deployment signal the moment accumulated 
data is sufficient, achieving certified deployment at 
$t^{*} = 2.6$~s from $130$~samples, $17.4$~seconds earlier 
than an offline batch baseline and with $16\times$ lower 
accumulated error during the unprotected disturbance window. 
The topology-aware estimator enforces known graph adjacency as 
exact zero constraints prior to regression, decoupling the 
estimation problem into independent per-node subproblems of 
size $\mathcal{O}(d_{\max})$ rather than $\mathcal{O}(N)$. On 
the G5 benchmark this reduces certification time by 59\% 
($1.62$~s versus $3.98$~s) and accumulated disturbance cost 
by 58\%, with a measured sample-count ratio of $2.46\times$ 
that grows to a projected $6$--$7\times$ at $n_{\mathrm{act}} 
= 10$ actuated nodes. The topology advantage is free (it 
requires only the network adjacency matrix, which is always 
available from engineering records) and it is exact, not 
an approximation.

Several limitations bound the current results. The contraction 
rate estimate $\hat{\beta}(t)$ relies on a linear 
integral-regression approximation of the closed-loop Jacobian, 
which is valid locally but may underestimate the true 
contraction rate in strongly nonlinear regimes far from the 
operating trajectory. The uncertainty radius $\rho(t)$ is 
derived empirically from residual statistics rather than from 
a formal noise model, so the certificate is conservative 
rather than tight. Both limitations will be addressed in 
future work through the two extensions described below.

Future work will extend the framework along two directions. First, we will develop 
\textit{distributionally robust} streaming certificates that 
remain valid across a Wasserstein ambiguity set of disturbance 
distributions, so that $\beta_{\mathrm{cert}}(t,\varepsilon)$ 
degrades gracefully rather than failing when the realized noise 
distribution shifts from Gaussian to heavy-tailed or impulsive (the three-regime structure that characterizes demand 
disturbances in water distribution and traffic networks.)
Second, we will develop a \textit{compositional} 
certification framework in which local streaming certificates 
from disconnected subsystems compose into a network-level 
stability guarantee via the small-gain condition 
$\gamma < \sqrt{\beta_A \cdot \beta_B}$, enabling 
certified control of coupled water-energy or 
arterial-freeway systems without a joint model of the 
interconnected network.

\bibliography{ieeeconf/IEEEfull}

@article{van2020data,
  title={Data informativity: A new perspective on data-driven analysis and control},
  author={Van Waarde, Henk J and Eising, Jaap and Trentelman, Harry L and Camlibel, M Kanat},
  journal={IEEE Transactions on Automatic Control},
  volume={65},
  number={11},
  pages={4753--4768},
  year={2020},
  publisher={IEEE}
}

@article{de2019formulas,
  title={Formulas for data-driven control: Stabilization, optimality, and robustness},
  author={De Persis, Claudio and Tesi, Pietro},
  journal={IEEE Transactions on Automatic Control},
  volume={65},
  number={3},
  pages={909--924},
  year={2019},
  publisher={IEEE}
}

@article{van2021beyond,
  title={Beyond persistent excitation: Online experiment design for data-driven modeling and control},
  author={van Waarde, Henk J},
  journal={IEEE Control Systems Letters},
  volume={6},
  pages={319--324},
  year={2021},
  publisher={IEEE}
}

@article{hu2025enforcing,
  title={Enforcing contraction via data},
  author={Hu, Zhongjie and De Persis, Claudio and Tesi, Pietro},
  journal={IEEE Transactions on Automatic Control},
  year={2025},
  publisher={IEEE}
}

@article{dorfler2022data,
  title={Data-enabled predictive control: In the shallows of the deepc},
  author={D{\"o}rfler, Florian and Berberich, Julian and K{\"o}hler, Johann and Allg{\"o}wer, Frank},
  journal={Annual Reviews in Control},
  volume={53},
  pages={123--142},
  year={2022}
}

@article{kohler2022data,
  title={Data-driven distributed MPC of dynamically coupled linear systems},
  author={Kohler, Matthias and Berberich, Julian and M{\"u}ller, Matthias A and Allgower, Frank},
  journal={IFAC-PapersOnLine},
  volume={55},
  number={30},
  pages={365--370},
  year={2022},
  publisher={Elsevier}
}

@article{lohmiller1998contraction,
  title={On contraction analysis for non-linear systems},
  author={Lohmiller, Winfried and Slotine, Jean-Jacques E},
  journal={Automatica},
  volume={34},
  number={6},
  pages={683--696},
  year={1998},
  publisher={Elsevier}
}

@article{brunton2016discovering,
  title={Discovering governing equations from data by sparse identification of nonlinear dynamical systems},
  author={Brunton, Steven L and Proctor, Joshua L and Kutz, J Nathan},
  journal={Proceedings of the national academy of sciences},
  volume={113},
  number={15},
  pages={3932--3937},
  year={2016},
  publisher={National Academy of Sciences}
}

@article{olfati2007consensus,
  title={Consensus and cooperation in networked multi-agent systems},
  author={Olfati-Saber, Reza and Fax, J Alex and Murray, Richard M},
  journal={Proceedings of the IEEE},
  volume={95},
  number={1},
  pages={215--233},
  year={2007},
  publisher={IEEE}
}

@article{willems2005note,
  title={A note on persistency of excitation},
  author={Willems, Jan C and Rapisarda, Paolo and Markovsky, Ivan and De Moor, Bart LM},
  journal={Systems \& Control Letters},
  volume={54},
  number={4},
  pages={325--329},
  year={2005},
  publisher={Elsevier}
}

@article{villacres2026data,
  title={Data-driven, model-free control for reliable operation of water distribution systems: Implementation, benchmarking and validation},
  author={Villacres, Daniela and Putri, Saskia and Moazeni, Faegheh},
  journal={Water Research},
  pages={126295},
  year={2026},
  publisher={Elsevier}
}

@article{moazeni2023data,
  title={Data-enabled identification of nonlinear dynamics of water systems using sparse regression technique},
  author={Moazeni, Faegheh and Khazaei, Javad},
  journal={IFAC-PapersOnLine},
  volume={56},
  number={2},
  pages={2389--2394},
  year={2023},
  publisher={Elsevier}
}

\end{document}